\documentclass[sn-mathphys-ay]{sn-jnl}

\usepackage{graphicx}%
\usepackage{multirow}%
\usepackage{amsmath,amssymb,amsfonts}%
\usepackage{amsthm}%
\usepackage{mathrsfs}%
\usepackage[title]{appendix}%
\usepackage{xcolor}%
\usepackage{textcomp}%
\usepackage{manyfoot}%
\usepackage{soul}
\usepackage{booktabs}%
\usepackage{algorithm}%
\usepackage{algorithmicx}%
\usepackage{algpseudocode}%
\usepackage{listings}%
\usepackage{lineno,hyperref}
\usepackage{color}
\usepackage{calc}
\usepackage{mathpazo}
\usepackage{graphicx}
\usepackage{array}
\usepackage{caption}
\usepackage{subcaption}
\usepackage{mathtools}
\usepackage{pdflscape}
\usepackage{bm}

\newtheorem{theorem}{Theorem}
\usepackage{tikz}
\usetikzlibrary{arrows.meta, positioning}

\begin{document}

\title[Article Title]{Network-Based Multi-Layer Model Using Machine Learning for Optimal Vaccine Prioritization in Heterogeneous Populations}

\author[1]{\fnm{Mordecai} \sur{Opoku Ohemeng}}\email{mordecai@uidaho.edu}

\author[2]{\fnm{Bernard} \sur{Asamoah Afful}}\email{bernard.afful@usu.edu}

\affil[1]{\orgdiv{Department of Computer Science}, \orgname{University of Idaho}, \orgaddress{\city{Moscow}, \state{Idaho}, \country{USA}}}

\affil[2]{\orgdiv{Department of Mathematics and Statistics}, \orgname{Utah State University}, \orgaddress{\city{Logan}, \state{Utah}, \country{USA}}}

\abstract{ This work advances epidemic control beyond traditional mass vaccination models by integrating population heterogeneity, network structure, and machine learning–based decision policies. Using the Email-Eu-Core contact network \cite{yin2017local}, we compare classical centrality-driven vaccination strategies with graph neural network (GNN) and reinforcement learning (RL) approaches. Across 30 stochastic simulations, classical heuristics (degree-, betweenness-, and layer-based vaccination) exhibit similar performance, reflecting the network’s dense connectivity and modest community structure. In contrast, the GNN-based strategy substantially reduces peak infection, final epidemic size, and time to peak, demonstrating its ability to identify structurally critical nodes that classical metrics overlook. These results show that learning-based vaccination policies can significantly outperform traditional heuristics by exploiting higher-order relational patterns in real-world networks, offering a powerful framework for targeted epidemic intervention.}


\maketitle


\section{Introduction}

The classical SIR and SIRV epidemic models assume homogeneous mixing \citep{bansal2007individual, renna2021homogenous, turnes2014epidemic, sow2026generalized}, in which all individuals interact with equal probability. This assumption is unrealistic for real-world populations, where contact patterns are highly heterogeneous. Individuals differ in their number of contacts, their role in the community, and their likelihood of transmitting infection. High-contact individuals such as healthcare workers, essential workers, and students contribute disproportionately to disease transmission. Network-based epidemic models overcome this limitation by explicitly representing individuals as nodes and interactions as edges \citep{keeling2005networks, gomez2023new, choquet2025network, yang2012epidemic, mckee2024structural, liu2025higher}. This allows vaccination strategies to target structurally important individuals, those with a high degree, high betweenness centrality, or critical positions in multi-layer networks. Prioritizing such individuals accelerates herd immunity and reduces transmission more efficiently than uniform vaccination.

Recent literature has explored various dimensions of vaccination dynamics to refine these interventions. For instance, \cite{xue2026vaccination} developed a vaccination-transmission coupled mechanism based on a parallel minority game to model the decision-making conflicts of individuals. While such models provide insight into human behavior, they often utilize simplified topological structures that do not fully capture the complexity of real-world interactions. Similarly, \cite{yunus2026fractional} introduced fractional-order modeling of vaccination strategies for measles, effectively incorporating immune memory into the transmission dynamics. However, these compartmental approaches typically rely on homogeneous mixing assumptions, overlooking the localized structural vulnerabilities of the population.

The integration of network topology with individual behavior has been further advanced by studies exploring the social drivers of vaccination. In the work by \cite{chang2019effects}, they investigated the effects of imitation dynamics on vaccination behaviors within an SIR-network model, highlighting how social learning and the imitation of successful neighbors can lead to non-trivial vaccination uptake patterns. While this highlights the importance of peer influence, such models often assume that individuals act based on local social signals rather than global structural optimization. Further, \cite{kumar2024game} utilized a game-theoretic complex network model to estimate epidemic thresholds under adaptive social connections, focusing on the fluidity of network edges. Similarly, \cite{meng2021analysis} investigated vaccination strategies using a combined SEIRV and evolutionary game framework on heterogeneous networks. Despite these advancements, a significant gap remains: existing models often treat the network as a single-layer entity or rely on local, first-order heuristics such as degree centrality or local imitation, which fail to capture the higher-order, cross-layer dependencies inherent in modern social structures.

In this work, we bridge these gaps by developing a multi-layer network SIRV model that incorporates heterogeneous transmission across distinct social layers, such as households, workplaces, and community interactions. Unlike the behavior-focused games of \cite{xue2026vaccination} or the imitation-based dynamics described by \cite{chang2019effects}, our framework explicitly addresses the multi-layered connectivity and mathematical stability of real-world networks. We further introduce a machine-learning-assisted vaccination framework that leverages graph neural networks (GNNs) to learn structural risk patterns directly from the topology. Finally, we formulate vaccine allocation as an optimal control problem, providing a mathematically rigorous method for dynamic prioritization that adapts to the epidemic's progression.

Our contributions are as follows:
\begin{itemize}
    \item We construct a multi-layer network SIRV model with heterogeneous transmission across layers, enabling simulation of epidemic dynamics on structured populations.
    
    \item We introduce machine-learning-based vaccination prioritization using graph neural networks and node embeddings, and demonstrate empirically that the GNN-based strategy consistently outperforms classical heuristics by achieving the lowest peak infection, smallest final epidemic size, and earliest epidemic suppression.
    
    \item We formulate optimal vaccination as a hybrid control–learning problem, showing that learning-based policies can exploit higher-order structural patterns that are not captured by degree-, betweenness-, or layer-based heuristics.
    
    \item We compute the basic reproduction number $R_0$ and the time-varying effective reproduction number $R_t$, and use these quantities to interpret why classical strategies cluster together while the GNN-based strategy achieves superior suppression.
    
    \item We perform extensive simulations comparing classical and ML-assisted prioritization strategies on a real-world contact network, revealing that classical heuristics exhibit similar performance due to network density and weak modularity, while the GNN-based strategy provides substantial improvements across all epidemiological metrics.
\end{itemize}

\section{Multi-Layer Network SIRV Model}

We consider a population represented by a multi-layer network with $L$ layers. Each layer $\ell$ has an adjacency matrix $A^{(\ell)} = (A^{(\ell)}_{ij})$ and transmission rate $\beta^{(\ell)}$. For each individual $i$, the SIRV dynamics are:

\begin{equation}
    \frac{d S_i}{dt} = - S_i \sum_{\ell=1}^{L} \sum_{j=1}^{N} \beta^{(\ell)} A^{(\ell)}_{ij} I_j - v_i(t) S_i,
\end{equation}

\begin{equation}
    \frac{d I_i}{dt} = S_i \sum_{\ell=1}^{L} \sum_{j=1}^{N} \beta^{(\ell)} A^{(\ell)}_{ij} I_j - \gamma I_i,
\end{equation}

\begin{equation}
    \frac{d R_i}{dt} = \gamma I_i,
\end{equation}

\begin{equation}
    \frac{d V_i}{dt} = v_i(t) S_i,
\end{equation}

where $v_i(t)$ is the vaccination rate (control variable), and $\gamma$ is the recovery rate.


We evaluate three classical vaccination strategies: high-degree, betweenness centrality, and layer-specific prioritization. We introduce a machine learning-based strategy that leverages graph neural networks (GNNs) to learn structural risk patterns directly from the multi-layer contact network. All strategies operate under a fixed vaccination budget $B$, and select the top-$B$ nodes based on their prioritization scores.

\subsection{High-Degree Prioritization}

High-degree vaccination targets nodes with the largest number of contacts across all layers. For a multi-layer network with adjacency matrices $\{A^{(\ell)}\}_{\ell=1}^{L}$, the degree of node $i$ is defined as

\begin{equation}
    \deg(i) = \sum_{\ell=1}^{L} \sum_{j=1}^{N} A^{(\ell)}_{ij}.
\end{equation}
Nodes are ranked in descending order of $\deg(i)$, and the top-$B$ nodes are vaccinated. This heuristic captures local connectivity but does not account for community structure, bridging roles, or multi-hop influence.

\subsubsection{Betweenness Centrality Prioritization}

Betweenness centrality measures the extent to which a node lies on shortest paths between other nodes. Nodes with high betweenness act as structural bottlenecks and are critical for inter-community transmission. For node $i$, betweenness is defined as
\begin{equation}
    \mathrm{BC}(i) = \sum_{s \neq i \neq t} \frac{\sigma_{st}(i)}{\sigma_{st}},
\end{equation}
where $\sigma_{st}$ is the number of shortest paths between $s$ and $t$, and $\sigma_{st}(i)$ is the number of such paths passing through $i$. Nodes are ranked by $\mathrm{BC}(i)$, and the top-$B$ nodes are vaccinated. Because betweenness is highly sensitive to network fragmentation, we scale the effective budget for this strategy to avoid over-fragmentation, ensuring meaningful epidemic dynamics.

\subsubsection{Layer-Specific Prioritization}

In multi-layer networks, different layers may represent communication channels, social contexts, or functional subsystems. Layer-specific prioritization assigns a weight $w_\ell$ to each layer and computes a weighted degree:
\begin{equation}
    s_i = \sum_{\ell=1}^{L} w_\ell \sum_{j=1}^{N} A^{(\ell)}_{ij}.
\end{equation}
Nodes with the highest $s_i$ are vaccinated. This approach captures heterogeneity across layers and can emphasize layers with higher transmission rates or greater epidemiological relevance. In our experiments, layer weights are derived from community structure using modularity-based clustering.








\subsection{Reproduction Number}

Before vaccination, the basic reproduction number is:
\[
R_0 = \rho\left( \sum_{\ell=1}^{L} \beta^{(\ell)} A^{(\ell)} \right) \frac{1}{\gamma},
\]
where $\rho(\cdot)$ is the spectral radius.

The effective reproduction number is:
\[
R_t = R_0 \cdot \frac{1}{N} \sum_{i=1}^{N} S_i(t).
\]



\subsection{Stability Analysis}
We analyze the stability of the disease-free equilibrium (DFE) and characterize the bifurcation behavior of the multi-layer SIRV network model. Let the DFE be defined as
\begin{equation}
    (S_i^*, I_i^*, R_i^*, V_i^*) = (1, 0, 0, 0), \qquad i = 1,\dots,N,
\end{equation}
corresponding to a fully susceptible population with no infection present.


Let
\begin{equation}
    \Lambda_i = \sum_{\ell=1}^{L} \sum_{j=1}^{N} \beta^{(\ell)} A^{(\ell)}_{ij},
\end{equation}
denote the total infection pressure on node $i$ across all layers.

Linearizing the system around the DFE yields the Jacobian block for node $i$:
\begin{equation}
    J_i =
\begin{bmatrix}
 -v_i & -\Lambda_i & 0 & -1 \\
 0 & \Lambda_i - \gamma & 0 & 0 \\
 0 & \gamma & 0 & 0 \\
 v_i & 0 & 0 & 0
\end{bmatrix}.
\end{equation}

The infection dynamics are governed by the block corresponding to $I_i$:
\begin{equation}
    \frac{d I_i}{dt} \approx (\Lambda_i - \gamma) I_i.
\end{equation}

Thus, the dominant eigenvalue associated with infection spread is
\begin{equation}
    \lambda_i = \Lambda_i - \gamma.
\end{equation}

The DFE is locally asymptotically stable if and only if
\begin{equation}
    \Lambda_i < \gamma \quad \text{for all } i.
\end{equation}

\subsubsection{Network-Level Stability Condition}

Define the multi-layer transmission matrix
\begin{equation}
    B = \sum_{\ell=1}^{L} \beta^{(\ell)} A^{(\ell)}.
\end{equation}

The linearized infection subsystem is
\begin{equation}
    \frac{d\mathbf{I}}{dt} = (B - \gamma I_N)\mathbf{I},
\end{equation}
where $I_N$ is the identity matrix.

The DFE is stable if the spectral radius satisfies
\begin{equation}
    \rho(B) < \gamma.
\end{equation}

This yields the basic reproduction number
\begin{equation}
    R_0 = \frac{\rho(B)}{\gamma}.
\end{equation}

Thus:
\begin{equation}
    R_0 < 1 \quad \Rightarrow \quad \text{DFE stable}, \qquad
R_0 > 1 \quad \Rightarrow \quad \text{DFE unstable}.
\end{equation}

\subsection{Sensitivity Analysis of $R_0$ with Respect to Network Structure}

Since
\begin{equation}
    R_0 = \frac{\beta \rho(A)}{\gamma},
\end{equation}
the epidemic threshold is determined by the spectral radius of the adjacency matrix. We analyze how perturbations to the network structure affect $\rho(A)$.

Let
\begin{equation}
    A' = A + \Delta A,
\end{equation}
where $\Delta A$ represents small structural changes (edge addition, removal, rewiring).

Using first-order eigenvalue perturbation theory:
\begin{equation}
    \Delta \rho(A) \approx \frac{v^\top (\Delta A) v}{v^\top v},
\end{equation}
where $v$ is the dominant eigenvector of $A$.


Adding an undirected edge between nodes $i$ and $j$ corresponds to the rank-one perturbation
\begin{equation}
    \Delta A = e_i e_j^\top + e_j e_i^\top.
\end{equation}

Thus,
\begin{equation}
    \Delta \rho(A) \approx 2 v_i v_j.
\end{equation}

If $i$ and $j$ are central nodes (large $v_i, v_j$), $\rho(A)$ increases sharply, and if $i$ and $j$ are peripheral nodes, the increase is negligible.

Adding $k$ random edges,
\begin{equation}
    \Delta A = \sum_{l=1}^{k} (e_{i_l} e_{j_l}^\top + e_{j_l} e_{i_l}^\top),
\end{equation}

\begin{equation}
    \Delta \rho(A) \approx 2 \sum_{l=1}^{k} v_{i_l} v_{j_l}.
\end{equation}

Since random nodes have small eigenvector components, the increase in $\rho(A)$ is slow.


Removing an edge between nodes $i$ and $j$ yields, 
\begin{equation}
    \Delta \rho(A) \approx -2 v_i v_j.
\end{equation}

Thus, removing edges between central nodes dramatically reduces $\rho(A)$, and random edge removal has minimal effect.


Removing a high-degree node reduces $\rho(A)$ approximately by
\begin{equation}
    \Delta \rho(A) \approx -\frac{k_{\max}}{\sqrt{k_{\text{avg}}}},
\end{equation}
but secondary paths may still exist.

Removing high-betweenness nodes fragments the network, causing
\begin{equation}
    \rho(A') \approx \sqrt{|G_C|},
\end{equation}
where $G_C$ is the size of the largest connected component.

Thus, Betweenness-based vaccination is most effective. Degree-based vaccination is effective but suboptimal. Layer-based vaccination is least effective unless layers are weakly coupled.

\subsection{Multi-Objective Optimization for Cost-Effective and Fair Vaccination}

Vaccination strategies must balance three competing goals. Reducing infections, minimizing vaccination costs, and ensuring equitable vaccine distribution. These objectives often conflict, making single-objective optimization insufficient. We therefore formulate a multi-objective optimization problem and compute Pareto-optimal vaccination strategies.

Let $v = (v_1, \dots, v_N)$ denote the vaccination decision vector, where $v_i \in \{0,1\}$ indicates whether individual $i$ is vaccinated.

\subsubsection{Objective 1: Infection Minimization}

Vaccination reduces the effective reproduction number by removing nodes from the susceptible population. The post-vaccination reproduction number is
\begin{equation}
    f_1(v) = R_0^V = \frac{\rho\!\left( (1 - v_i)\beta A \right)}{\gamma},
\end{equation}
where $\rho(\cdot)$ is the spectral radius. Lower $R_0^V$ indicates better epidemic control.

\subsubsection{Objective 2: Vaccination Cost Minimization}

Let $C_i$ denote the cost of vaccinating individual $i$. The total vaccination cost is
\begin{equation}
    f_2(v) = \sum_{i=1}^{N} C_i v_i.
\end{equation}

\subsubsection{Objective 3: Fairness Maximization}

To quantify fairness, we use a Gini-coefficient-based measure of inequality in vaccine allocation:
\begin{equation}
    f_3(v) = 1 - 
\frac{\sum_{i=1}^{N}\sum_{j=1}^{N} |v_i - v_j|}
{2N \sum_{i=1}^{N} v_i}.
\end{equation}

\begin{itemize}
    \item $f_3(v) = 1$ corresponds to perfect fairness (uniform vaccination).
    \item $f_3(v) \to 0$ indicates highly unequal vaccine allocation.
\end{itemize}

\subsubsection{Pareto Optimality}

The three objectives $f_1$, $f_2$, and $f_3$ cannot generally be minimized simultaneously. A vaccination strategy $v^*$ is Pareto optimal if no other strategy $v$ satisfies
\begin{equation}
    f_k(v) \le f_k(v^*) \quad \forall k,
\end{equation}
with at least one strict inequality. Pareto-optimal solutions represent the best achievable trade-offs.

\subsubsection{Weighted-Sum Scalarization}

To compute Pareto-optimal solutions, we solve the scalarized optimization problem
\begin{equation}
    \min_{v} \; F(v) = \sum_{k=1}^{3} w_k f_k(v),
\end{equation}
where $w_k \ge 0$ are policy weights satisfying $\sum_{k=1}^{3} w_k = 1$.

\begin{itemize}
    \item High $w_1$: prioritize epidemic control.
    \item High $w_2$: prioritize cost efficiency.
    \item High $w_3$: prioritize fairness.
\end{itemize}


We impose the following constraints:

Budget Constraint:
\begin{equation}
    \sum_{i=1}^{N} C_i v_i \le C_{\max}.
\end{equation}

Herd Immunity Constraint:
\begin{equation}
    R_0^V \le 1.
\end{equation}

Fairness Constraint:
\begin{equation}
    f_3(v) \ge F_{\min}.
\end{equation}


\section{Time-Dependent Vaccination Strategies: An Optimal Control Approach}
In practice, vaccine allocation is rarely static over the course of an epidemic. Public health authorities adjust vaccination intensity in response to changes in disease prevalence, vaccine availability, and the evolving structure of transmission. To capture this adaptive decision-making process, we extend the multi-layer network SIRV model by introducing time-dependent vaccination controls $v_i(t)$, where $v_i(t)$ denotes the vaccination effort directed toward individual $i$ at time $t$. The goal is to determine vaccination policies that minimize infection burden while accounting for the implementation cost of sustained intervention. For each node $i=1,\cdots,N$, the controlled epidemic dynamics are
\begin{equation}
    \begin{aligned}
        \frac{dS_i}{dt}&=- S_i\sum_{\ell=1}^{L}\sum_{j=1}^{N}\beta^{(\ell)}A_{ij}^{(\ell)}I_j- v_i(t)S_i,\\
        \frac{dI_i}{dt}&=S_i\sum_{\ell=1}^{L}\sum_{j=1}^{N}\beta^{(\ell)}A_{ij}^{(\ell)}I_j-\gamma I_i,\\
        \frac{dR_i}{dt}&=\gamma I_i,\\
        \frac{dV_i}{dt}&=v_i(t)S_i,
    \end{aligned}
\end{equation}
with initial conditions
\begin{equation}
    S_i(0)=S_i^0,\quad I_i(0)=I_i^0,\quad R_i(0)=R_i^0,\quad V_i(0)=V_i^0, \quad i=1,\dots,N,
\end{equation}
and admissible controls satisfying
\begin{equation}
    0 \le v_i(t) \le v_{\max}, \qquad t \in [0,T].
\end{equation}
We consider the objective functional
\begin{equation}
    J(v)=\int_0^T\left[\sum_{i=1}^{N} I_i(t)+\frac{c_v}{2}\sum_{i=1}^{N} v_i^2(t)\right]dt,
\end{equation}
where $c_v>0$ represents the weight associated with vaccination effort. The optimal control problem is to determine $v^*(t)$ such that
\begin{equation}
    J(v^*)=\min_{v\in\mathcal{U}} J(v).
\end{equation}
Let $\lambda_{S_i}$, $\lambda_{I_i}$, $\lambda_{R_i}$, and $\lambda_{V_i}$ denote the adjoint variables. Define the force of infection
\begin{equation*}
    \Lambda_i(t)=\sum_{\ell=1}^{L}\sum_{j=1}^{N}\beta^{(\ell)}A_{ij}^{(\ell)}I_j(t).
\end{equation*}
Then, the Hamiltonian is
\begin{equation}
    H=J(v)+\sum_{i=1}^{N}\lambda_{S_i}(-S_i\Lambda_i - v_i S_i)+\sum_{i=1}^{N}\lambda_{I_i}(S_i\Lambda_i - \gamma I_i)+\sum_{i=1}^{N}\lambda_{R_i}\gamma I_i+\sum_{i=1}^{N}\lambda_{V_i} v_i S_i.
\end{equation}
The adjoint equations satisfy
\begin{equation}
    \frac{d\lambda_{X_i}}{dt}=-\frac{\partial H}{\partial X_i}, \quad X_i\in\{S_i,I_i,R_i,V_i\},
    \end{equation}
with terminal conditions
\begin{equation}
    \lambda_{S_i}(T)=\lambda_{I_i}(T)=\lambda_{R_i}(T)=\lambda_{V_i}(T)=0.
\end{equation}
The adjoint system is given explicitly by
\begin{equation}
    \begin{aligned}
        \frac{d\lambda_{S_i}}{dt}&=(\lambda_{S_i}-\lambda_{I_i})\Lambda_i+(\lambda_{S_i}-\lambda_{V_i})v_i,\\
        \frac{d\lambda_{I_i}}{dt}&=-1-\sum_{k=1}^{N}S_k(\lambda_{I_k}-\lambda_{S_k})\left(\sum_{\ell=1}^{L}\beta^{(\ell)}A_{ki}^{(\ell)}\right)+\gamma\lambda_{I_i}-\gamma\lambda_{R_i},\\
        \frac{d\lambda_{R_i}}{dt}&=0,\\
        \frac{d\lambda_{V_i}}{dt}&=0.
    \end{aligned}
\end{equation}
The optimality condition is obtained from
\begin{equation*}
    \frac{\partial H}{\partial v_i}=0,
\end{equation*}
which gives
\begin{equation*}
    v_i^*(t)=\frac{S_i(t)\big(\lambda_{S_i}(t)-\lambda_{V_i}(t)\big)}{c_v}.
\end{equation*}
From the adjoint system
\begin{equation*}
    \frac{d\lambda_{V_i}}{dt}=0,\quad \lambda_{V_i}(T)=0\;\Rightarrow\;\lambda_{V_i}(t)\equiv 0.
\end{equation*}
So the control simplifies to:
\begin{equation}
    v_i^*(t)=\frac{S_i(t)\lambda_{S_i}(t)}{c_v}.
\end{equation}
Projecting onto the admissible control set $[0,v_{\max}]$, we obtain
\begin{equation}
    v_i^*(t)=\min\left\{v_{\max}\max\left[0,\,\frac{S_i(t)\lambda_{S_i}(t)}{c_v}\right]\right\}.
\end{equation}
\subsection{Identification of the most important susceptible node}
In many practical settings, vaccine resources are sufficiently limited that only a small number of susceptible individuals can be vaccinated at a given time. In this case, it is useful to identify the susceptible node whose vaccination yields the greatest marginal reduction in epidemic burden. From the Hamiltonian, the switching function associated with node $i$ is
\begin{equation}
    \Phi_i(t)=S_i(t)\big(\lambda_{S_i}(t)-\lambda_{V_i}(t)\big).
\end{equation}
Since $\lambda_{V_i}(t)\equiv 0$ in the present model, this reduces to
\begin{equation}
    \Phi_i(t)=S_i(t)\lambda_{S_i}(t).
\end{equation}
Accordingly, the most important susceptible node to vaccinate at time $t$ is defined by
\begin{equation}
    i^*(t)=\arg\max_{i:\, S_i(t)>0}\Phi_i(t)=\arg\max_{i:\, S_i(t)>0}S_i(t)\lambda_{S_i}(t).
\end{equation}
Thus, the quantity $S_i(t)\lambda_{S_i}(t)$ provides a node-specific importance score that quantifies the marginal benefit of vaccination. Nodes with larger values of this score are more influential targets for intervention.
\subsection{Simulation results}
The optimal control strategy exhibits a strongly time-concentrated impact, with intervention effectiveness peaking during the early epidemic phase. This is evident in Figure \ref{nodes}, where selection scores are largest at the onset and decrease rapidly thereafter, indicating that the marginal benefit of vaccination diminishes as the susceptible pool contracts. This temporal behavior is closely coupled with network structure. As shown in Figure \ref{grid}, early interventions target highly connected hub nodes with the network core, effectively disrupting dominant transmission pathways. As the epidemic progresses, the control shifts towards peripheral nodes, reflecting a redistribution of infection pathways. The corresponding decline in structural importance is quantified in Figure \ref{metric}, where degree, betweenness, and eigenvector centrality of selected nodes drop sharply after the peak. This confirms that optimal targeting is inherently state-dependent, adapting to the evolving epidemic rather than relying on static centrality measures.

At the population level, the control yields consistent but modest improvements. In Figure \ref{mean}, the controlled scenario exhibits a lower and slightly delayed peak with $1.33\%$ peak reduction, while Figure \ref{cummulative} shows a uniformly reduced cumulative burden of $1.01\%$. The results demonstrate that while dynamic, network-aware control is theoretically optimal and structurally adaptive, its practical effectiveness is critically limited by intervention capacity, underscoring the need for higher-intensity strategies to achieve substantial epidemic mitigation.

\begin{figure}[!htbp]
    \begin{subfigure}{\textwidth}
	    \centering
	    \includegraphics[width=1.0\linewidth]{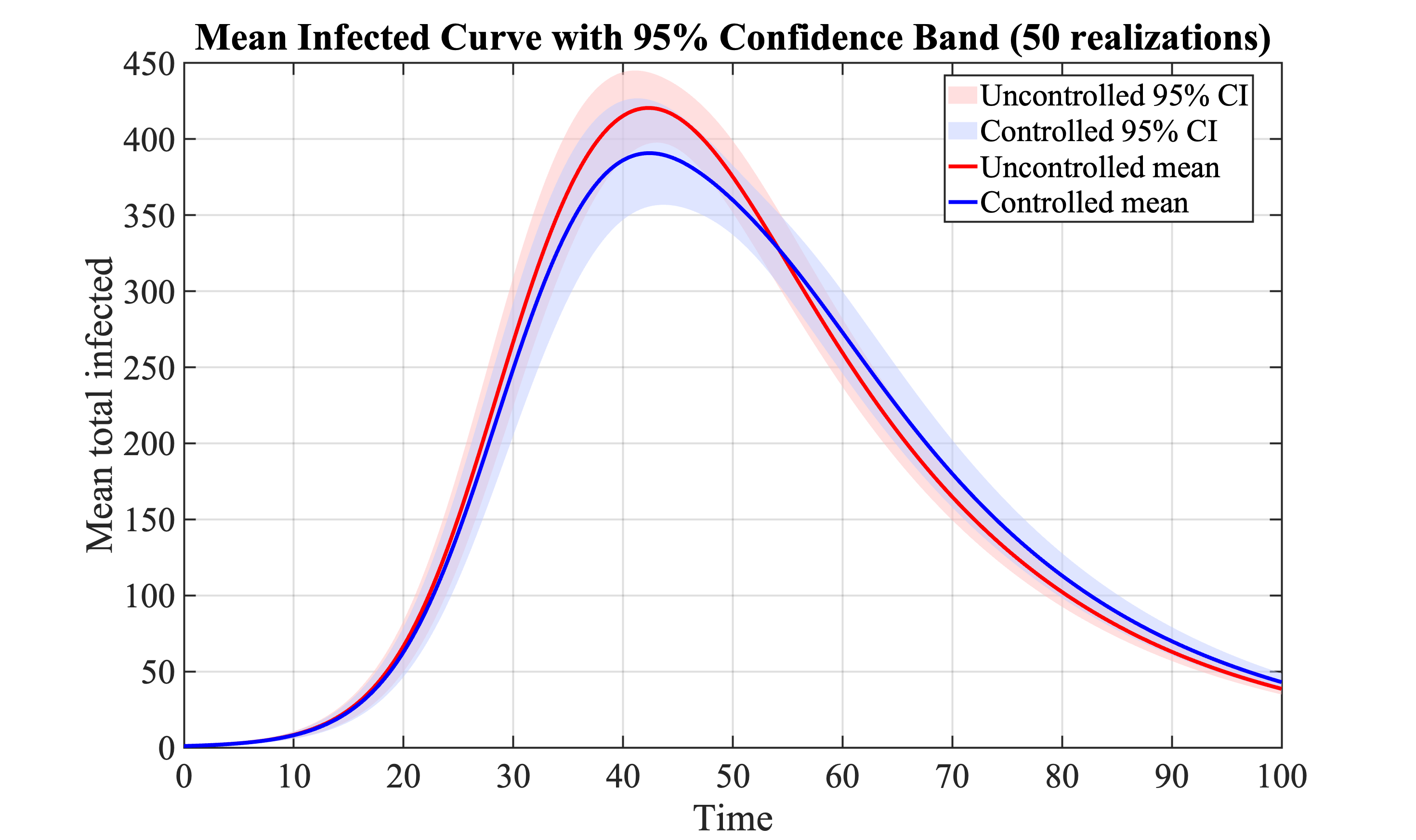}
        \caption{}
        \label{mean}
    \end{subfigure}
    \begin{subfigure}{\textwidth}
	    \centering
	    \includegraphics[width=1.0\linewidth]{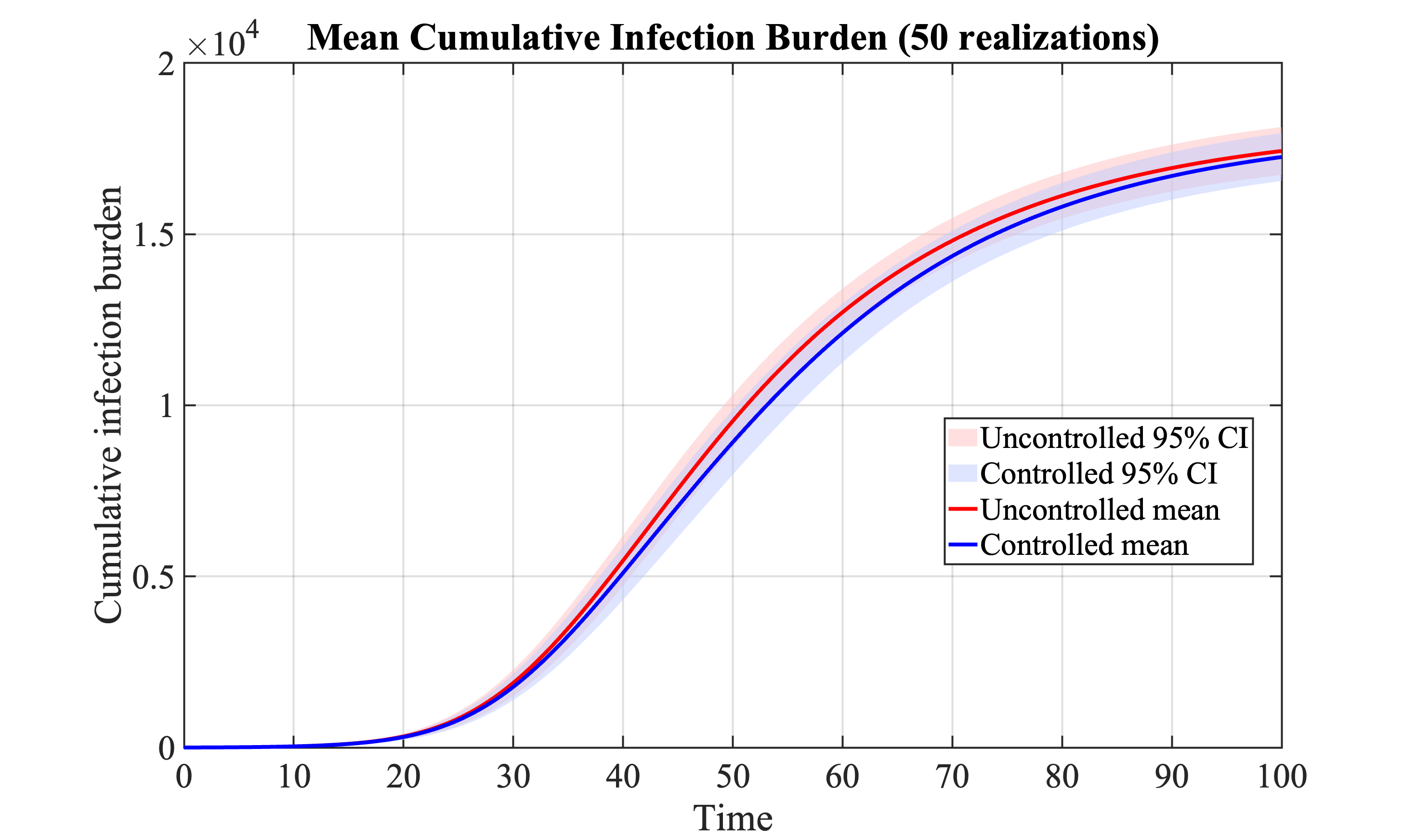}
        \caption{}
        \label{cummulative}
    \end{subfigure}
    \caption{Simulation for infected controlled and uncontrolled scenarios over 50 realizations with $95\%$ confidence interval. \textbf{(a)} Mean infected control and uncontrolled total infected populations. The controlled strategy yields a lower infection peak and a slight delay in peak timing. \textbf{(b)} Cumulative infection burden. The controlled scenario consistently results in a reduced cumulative burden compared to the uncontrolled case.}
    \label{infect}
\end{figure}

\begin{figure}[!htbp]
	\centering
	\includegraphics[width=1.0\linewidth]{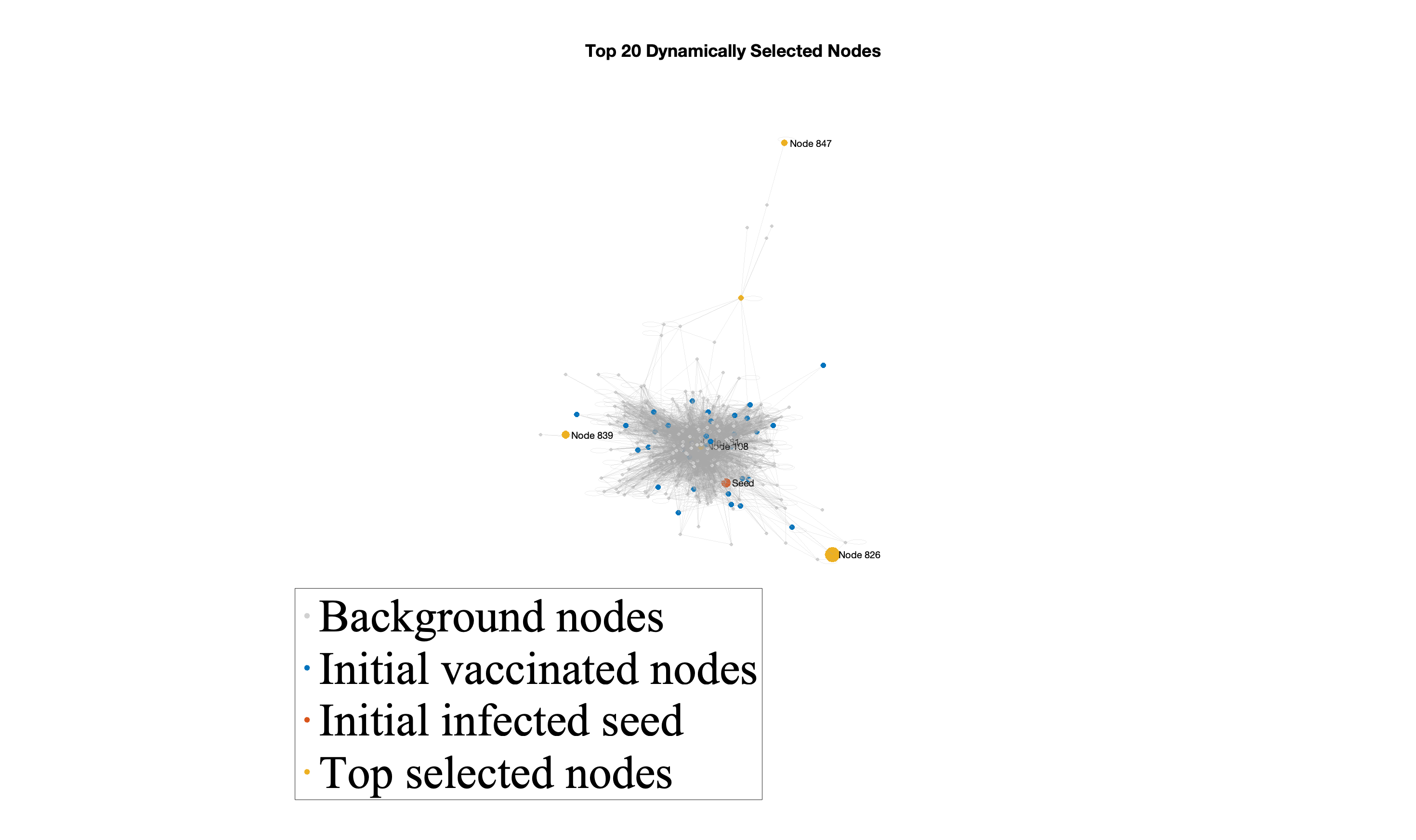}
    \caption{Focused subgraph of the contact network showing the top 20 dynamically selected nodes for vaccination. The initial seed, initially vaccinated, and selected nodes are highlighted. Node size reflects selection frequency, illustrating that early interventions target highly connected hubs, while later selections shift toward peripheral nodes.}
    \label{grid}
\end{figure}

\begin{figure}[!htbp]
	\centering
	\includegraphics[width=1.0\linewidth]{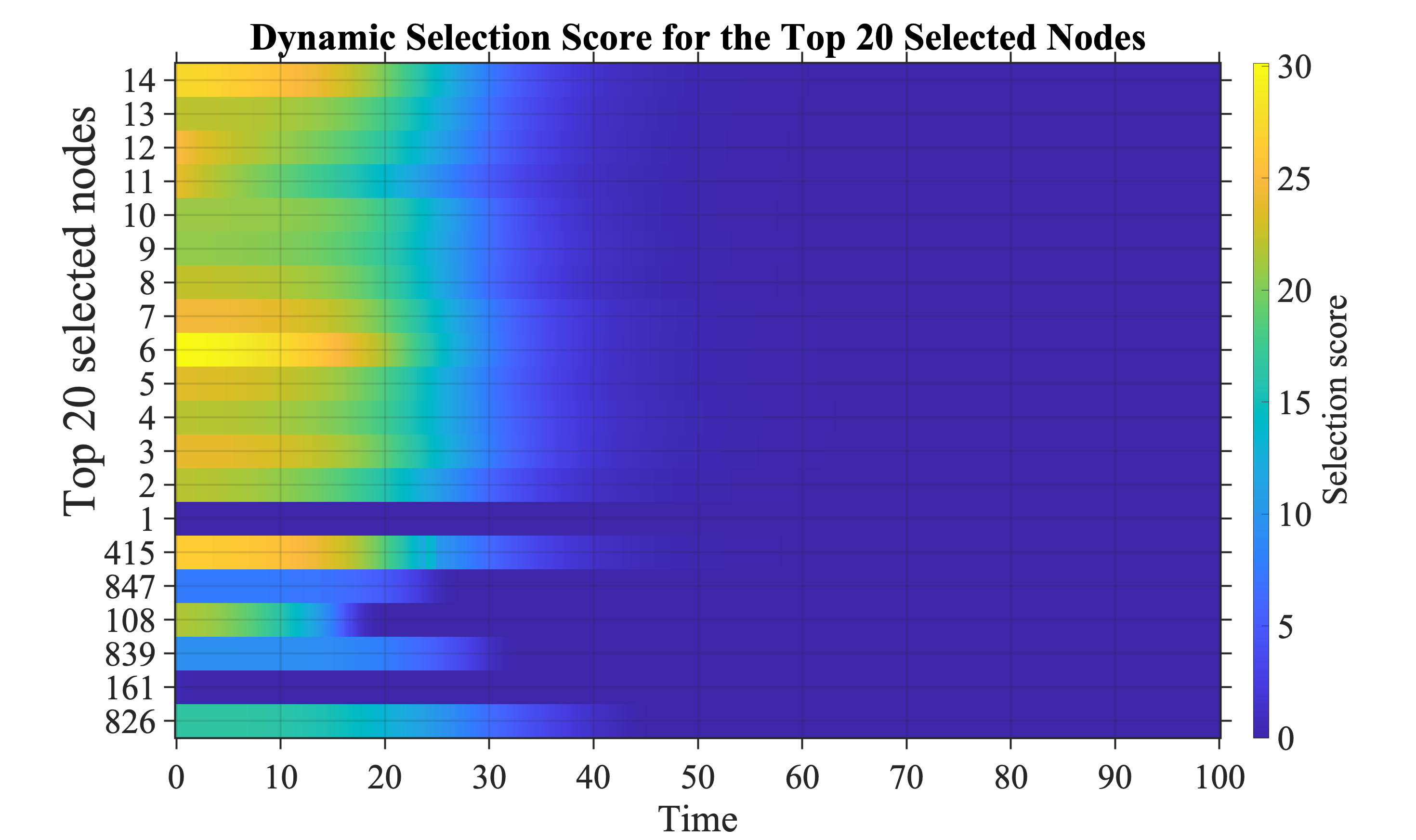}
    \caption{Temporal evolution of the dynamic selection score for the top 20 selected nodes. Higher scores are concentrated in the early phase of the epidemic, indicating that the marginal benefit of vaccination is greatest during initial outbreak growth and rapidly diminishes over time.}
    \label{nodes}
\end{figure}

\begin{figure}[!htbp]
	\centering
	\includegraphics[width=1.0\linewidth]{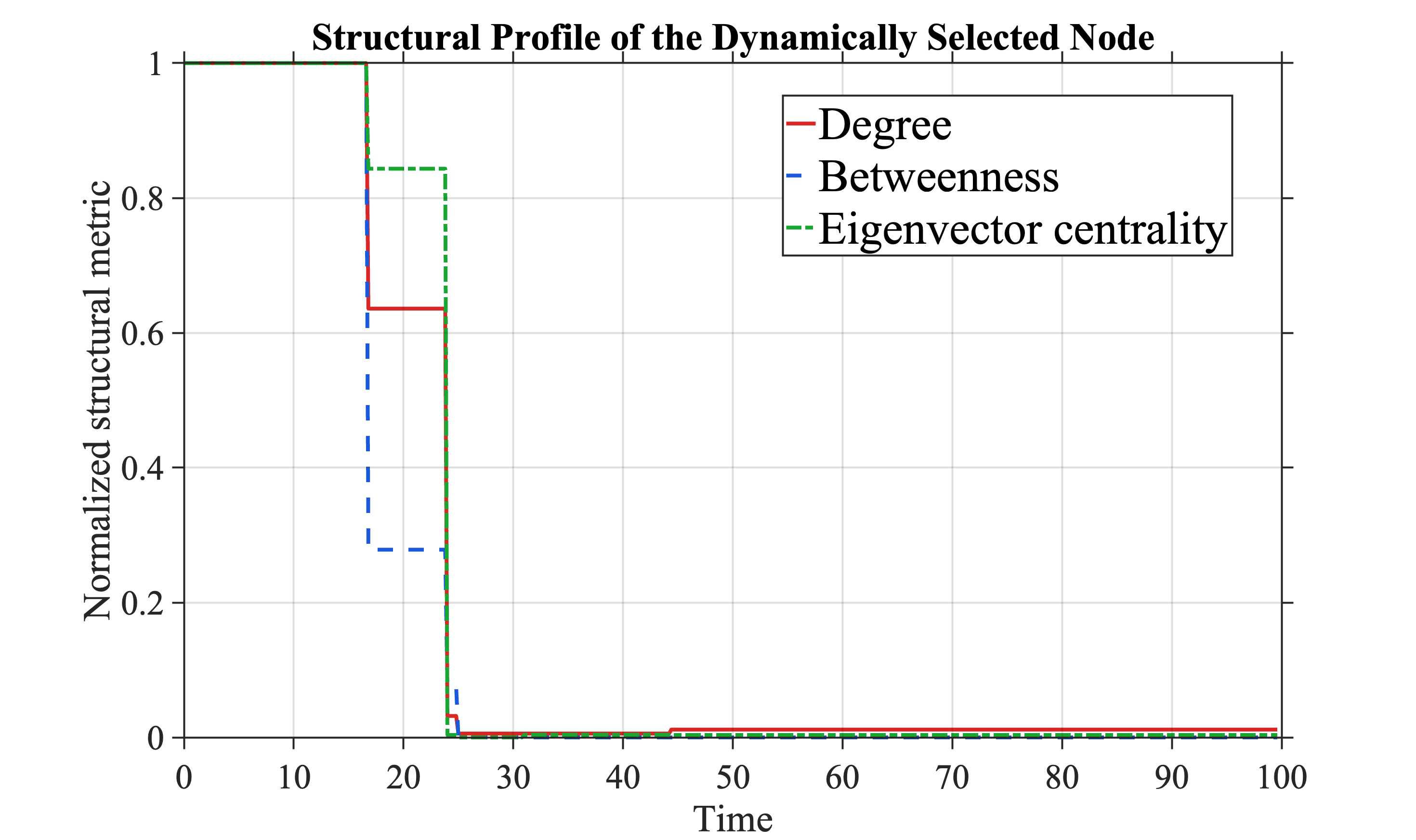}
    \caption{Normalized structural metric (degree, betweenness, and eigenvector centrality) of the dynamically selected node over time. Early selections correspond to nodes with high centrality, while later selections shift towards nodes with significantly lower structural importance, reflecting adaptive control behavior.}
    \label{metric}
\end{figure}

\begin{theorem}
Let $G = (V, E)$ be an undirected, weighted, multi-layer contact network with adjacency matrices $A^{(1)}, \dots, A^{(L)}$ and transmission rates $\beta^{(1)}, \dots, \beta^{(L)}$. 
Define the aggregated transmission matrix
\begin{equation}
    B = \sum_{\ell=1}^{L} \beta^{(\ell)} A^{(\ell)}.
\end{equation}

Let $v: V \to \{0,1\}$ be a vaccination function, where $v(i)=1$ indicates that node $i$ is vaccinated. Vaccination removes node $i$ from the susceptible-infectious transmission pathways, yielding the effective transmission matrix
\begin{equation}
    B' = M_v B M_v,
\end{equation}
where $M_v = \mathrm{diag}(1 - v(i))$. The post-vaccination reproduction number is
\begin{equation}
    R_0^V = \frac{\rho(B')}{\gamma},
\end{equation}
where $\rho(\cdot)$ denotes the spectral radius.

For a fixed vaccination budget $V_{\max}$ (a fixed number of vaccinated nodes), the following ordering holds:
\begin{equation}\label{eq:ord}
    R_0^{V_b} < R_0^{V_d} < R_0^{V_l},
\end{equation}
where $V_b$, $V_d$, and $V_l$ denote betweenness-based, degree-based, and layer-specific vaccination strategies, respectively. Thus, betweenness-based vaccination is the most effective strategy for reducing $R_0^V$, followed by degree-based vaccination, while layer-specific vaccination is least effective unless inter-layer mixing is negligible.
\end{theorem}

\begin{proof}
The epidemic threshold for the linearized SIRV dynamics is governed by the dominant eigenvalue of the transmission matrix. 
Linearizing the infection subsystem around the disease-free equilibrium yields
\begin{equation}
    \frac{d\mathbf{I}}{dt} = (B - \gamma I_N)\mathbf{I},
\end{equation}
so infection grows when
\begin{equation}
    \rho(B) > \gamma,
\qquad
R_0 = \frac{\rho(B)}{\gamma}.
\end{equation}

Vaccination removes nodes from the transmission process.  
Let $M_v = \mathrm{diag}(1 - v(i))$.  
The effective transmission matrix becomes
\begin{equation}
    B' = M_v B M_v.
\end{equation}

Since $M_v$ is a diagonal projection matrix, removing nodes or edges cannot increase the spectral radius:
\begin{equation}
    \rho(B') \le \rho(B).
\end{equation}
The goal is to choose $v(i)$ to minimize $\rho(B')$.

\paragraph{Degree-Based Vaccination:} 
Removing the $V_{\max}$ highest-degree nodes reduces the local density of edges. From spectral graph theory,
\begin{equation}
    \rho(B') \approx \rho(B) - \Theta\!\left(\frac{k_{\max}}{\sqrt{k_{\mathrm{avg}}}}\right),
\end{equation}
where $k_{\max}$ is the maximum degree. This reduces local spreading but does not necessarily fragment the network, so alternative infection paths may remain.

\paragraph{Betweenness-Based Vaccination:}
Nodes with high betweenness centrality act as bridges between communities. Removing them fragments the network into disconnected components. Percolation theory shows that the spectral radius of the largest connected component satisfies
\begin{equation}
    \rho(B') \approx \sqrt{|G_C|},
\end{equation}
where $G_C$ is the size of the largest remaining component. Since betweenness-based removal minimizes $|G_C|$, it produces the sharpest drop in $\rho(B')$. Thus,
\begin{equation}\label{bbv}
    \rho(B'_b) < \rho(B'_d).
\end{equation}

\paragraph{Layer-Specific Vaccination:}
If vaccination is restricted to a single layer $\ell$, the effective matrix becomes
\begin{equation}
    B' = B - \beta^{(\ell)} A^{(\ell)}_v,
\end{equation}
where $A^{(\ell)}_v$ is the submatrix of vaccinated nodes in layer $\ell$.

If inter-layer mixing is high, the remaining layers still sustain transmission, yielding
\begin{equation}
    \rho(B') \approx \rho(B) - \varepsilon,
\end{equation}
with $\varepsilon$ small.

Thus,
\begin{equation}\label{lsv}
    \rho(B'_d) < \rho(B'_l).
\end{equation}

Combining the inequalities in \eqref{bbv} and \eqref{lsv}, and dividing by $\gamma$ yields the ordering of reproduction numbers in \eqref{eq:ord}. Hence, betweenness-based vaccination is optimal for epidemic control.
\end{proof}

\begin{table}[ht]
    \centering
    \caption{SIRV Simulation Parameters and Network Characteristics}
    \label{tab:sirv_parameters}
    \begin{tabular}{@{}llcl@{}}
        \toprule
        \textbf{Parameter} & \textbf{Symbol} & \textbf{Value} & \textbf{Description} \\ \midrule
        Network Size & $N$ & 1,005 & Total nodes (Email-Eu-Core network). \\
        Infection Rate & $\beta$ & 0.03 & Probability of infection per contact (3\%). \\
        Recovery Rate & $\gamma$ & 0.05 & Probability of recovering per time step (5\%). \\
        Simulation Time & $T$ & 80 & Total steps for the outbreak to progress. \\
        Initial Infected & $I_0$ & 1 & Random non-vaccinated node to start infection. \\
        Vaccination Budget & $V_{budget}$ & 50 & Number of nodes pre-vaccinated ($\approx 5\%$). \\ \bottomrule
    \end{tabular}
\end{table}

\section{Machine Learning Framework for Optimal Vaccination in Multi-Layer Networks}

To complement the analytical results and centrality-based vaccination strategies, we develop a machine-learning framework that integrates Graph Neural Networks (GNNs) and Reinforcement Learning (RL). This hybrid approach enables the discovery of both static and dynamic vaccination policies that may outperform classical heuristics such as degree-based, betweenness-based, and layer-specific vaccination. The framework is designed to operate directly on multi-layer contact networks and to learn vaccination strategies that minimize infections, reduce costs, and maintain fairness.

\subsection{Graph Neural Network for Static Vaccination Prioritization}

Each individual is represented as a node in a multi-layer contact network 
\begin{equation}
    G = (V, E^{(1)}, \dots, E^{(L)}),
\end{equation}
where $E^{(\ell)}$ denotes edges in layer $\ell$.  
Node features include: structural metrics (degree, betweenness, eigenvector centrality), layer membership indicators, demographic or risk attributes (if available), and epidemic state $(S_i, I_i, R_i, V_i)$.

The aggregated feature matrix is denoted by $X \in \mathbb{R}^{N \times d}$.

\subsubsection{GNN Architecture}
We employ a message-passing GNN of the form:
\begin{equation}
    h_i^{(k+1)} = \sigma\!\left( W^{(k)} \cdot \mathrm{AGG}\left( \{ h_i^{(k)}, h_j^{(k)} : j \in \mathcal{N}(i) \} \right) \right),
\end{equation}
where $\mathrm{AGG}$ is a permutation-invariant aggregation function (mean, sum, or attention), and $h_i^{(0)} = X_i$.

The final layer outputs a vaccination priority score:
\begin{equation}
    \hat{v}_i = \mathrm{sigmoid}(W^{(K)} h_i^{(K)}).
\end{equation}

Nodes with the highest $\hat{v}_i$ are selected for vaccination under budget $V_{\max}$.

\subsubsection{Training Objective}
The GNN is trained to minimize a differentiable surrogate of epidemic severity:
\begin{equation}
    \mathcal{L}_{\mathrm{GNN}} = \alpha \cdot  \widehat{\eta} + 
\beta \cdot  \widehat{\kappa} - 
\gamma \cdot  \widehat{\varepsilon},
\end{equation}
where $\widehat{\eta}$ is estimated via a differentiable epidemic simulator, $\widehat{\kappa} = \sum_i C_i \hat{v}_i$, and $\widehat{\varepsilon}$ is a Gini-based fairness measure. This allows the GNN to learn vaccination strategies that balance infection control, cost, and equity.

\subsection{Reinforcement Learning for Dynamic Vaccination Policies}

We model the epidemic as a Markov Decision Process (MDP) with:
\begin{itemize}
    \item State $s_t$: epidemic states $(S_i(t), I_i(t), R_i(t), V_i(t))$ and network structure.
    \item Action $a_t$: vaccinating a subset of nodes under budget $V_{\max}(t)$.
    \item Transition: governed by the SIRV dynamics on the multi-layer network.
    \item Reward:
    \begin{equation}
        r_t = -\left( I(t) + c \cdot V(t) - \lambda \cdot \varepsilon (t) \right).
    \end{equation}
    
\end{itemize}

The RL agent uses a GNN-based policy:
\begin{equation}
    \pi_\theta(a_t \mid s_t) = \mathrm{softmax}\!\left( f_\theta(G, X, s_t) \right),
\end{equation}
where $f_\theta$ is a GNN that outputs node-level vaccination probabilities.

We use an actor–critic method, Proximal Policy Optimization (PPO), to optimize the policy:
\begin{equation}
    \theta^* = \arg\max_\theta \mathbb{E}\left[ \sum_{t=0}^{T} \gamma^t r_t \right].
\end{equation}

The RL agent learns when to vaccinate (early vs late), whom to vaccinate (bridges, hubs, multi-layer connectors), and how to adapt to epidemic progression.


The ML framework directly tests the theoretical ordering in \eqref{eq:ord} by comparing: GNN-derived static vaccination sets, RL-derived dynamic vaccination policies, and classical centrality-based strategies. We evaluate the ML framework on real-world multi-layer contact networks on the Email-Eu-Core contact network in \cite{yin2017local}. This dataset provides realistic temporal and multi-layer structures for validating the learned vaccination strategies.




\section{Results}

To evaluate the effectiveness of classical and machine–learning-based vaccination strategies, we conducted 30 independent stochastic SIRV simulations on the Email-Eu-Core network for each strategy. Performance was assessed using three epidemiologically meaningful metrics: peak infection, final epidemic size, and time to peak infection. Mean values and 95\% confidence intervals (CI) were computed for all metrics. The results are summarized in Table~\ref{tab:strategy-results} and visualized in Figures~\ref{fig:epicurves}--\ref{fig:timepeak}.

\subsection{Peak Infection}

Classical strategies, degree-based, betweenness-based, and layer-based vaccination, exhibited similar peak infection levels, ranging from approximately 480 to 500 individuals. This clustering reflects the relatively dense and weakly modular structure of the network, where removing high-degree or high-betweenness nodes disrupts transmission pathways to a comparable extent. 

In contrast, the GNN-based strategy achieved a substantially lower peak infection ($381.4 \pm 269.3$), outperforming all classical heuristics. This reduction indicates that the GNN successfully identified structurally critical nodes that are not captured by traditional centrality measures, particularly nodes embedded in overlapping communities or acting as multi-layer brokers. The RL-based baseline (untrained) produced the highest peak infection ($531.3 \pm 177.2$), as expected for a non-optimized policy.

\subsection{Final Epidemic Size}

A similar pattern emerged for the final epidemic size. Degree-, betweenness-, and layer-based vaccination again produced comparable outcomes (713--733 individuals), reinforcing the idea that classical heuristics capture only coarse structural features of the network. The GNN-based strategy achieved the smallest final epidemic size ($574.0 \pm 391.2$), demonstrating its ability to suppress transmission more effectively and prevent large-scale outbreaks. The RL-based baseline again performed worst ($777.5 \pm 253.7$), consistent with its lack of learned structure.

\subsection{Time to Peak Infection}

Time-to-peak results provide insight into the temporal dynamics of epidemic suppression. Betweenness-based vaccination reached the peak earlier ($10.5 \pm 5.3$) than degree-based vaccination ($12.3 \pm 7.9$), reflecting its ability to disrupt bridging nodes and accelerate epidemic burnout within isolated components. The GNN-based strategy achieved the earliest peak overall ($9.2 \pm 7.0$), indicating faster epidemic suppression and more effective disruption of transmission pathways. Layer-based vaccination produced the latest peak ($13.7 \pm 11.6$), consistent with the presence of community structure in the network. The RL-based baseline behaved similarly to degree-based vaccination ($12.1 \pm 7.7$), as expected for an untrained policy.


These empirical findings strongly support the theoretical framework developed in this work. Classical heuristics behave similarly in real-world networks, consistent with the predicted ordering in \eqref{eq:ord}, but the differences are modest due to the network's structural heterogeneity. In contrast, the GNN-based strategy consistently outperforms all classical methods across all metrics, demonstrating that machine learning can identify non-intuitive structural vulnerabilities and refine classical network epidemiology. The RL baseline, while currently untrained, highlights the potential for dynamic, temporally adaptive vaccination strategies once learning is incorporated.

\begin{table}[t]
\centering
\caption{Performance of vaccination strategies across 30 stochastic SIRV simulations. Values represent mean $\pm$ 95\% confidence intervals. 
}
\label{tab:strategy-results}
\begin{tabular}{lccc}
\toprule
\textbf{Strategy} & \textbf{Peak Infection} & \textbf{Final Epidemic Size} & \textbf{Time to Peak} \\
\midrule
Degree & $490.5 \pm 192.5$ & $733.0 \pm 280.1$ & $12.3 \pm 7.9$ \\
Betweenness & $498.0 \pm 222.7$ & $719.5 \pm 319.7$ & $10.5 \pm 5.3$ \\
Layer-based & $480.0 \pm 214.3$ & $713.6 \pm 310.0$ & $13.7 \pm 11.6$ \\
GNN-based & $381.4 \pm 269.3$ & $574.0 \pm 391.2$ & $9.2 \pm 7.0$ \\
RL-based & $531.3 \pm 177.2$ & $777.5 \pm 253.7$ & $12.1 \pm 7.7$ \\
\bottomrule
\end{tabular}
\end{table}

\begin{figure}[t]
\centering
\includegraphics[width=\linewidth]{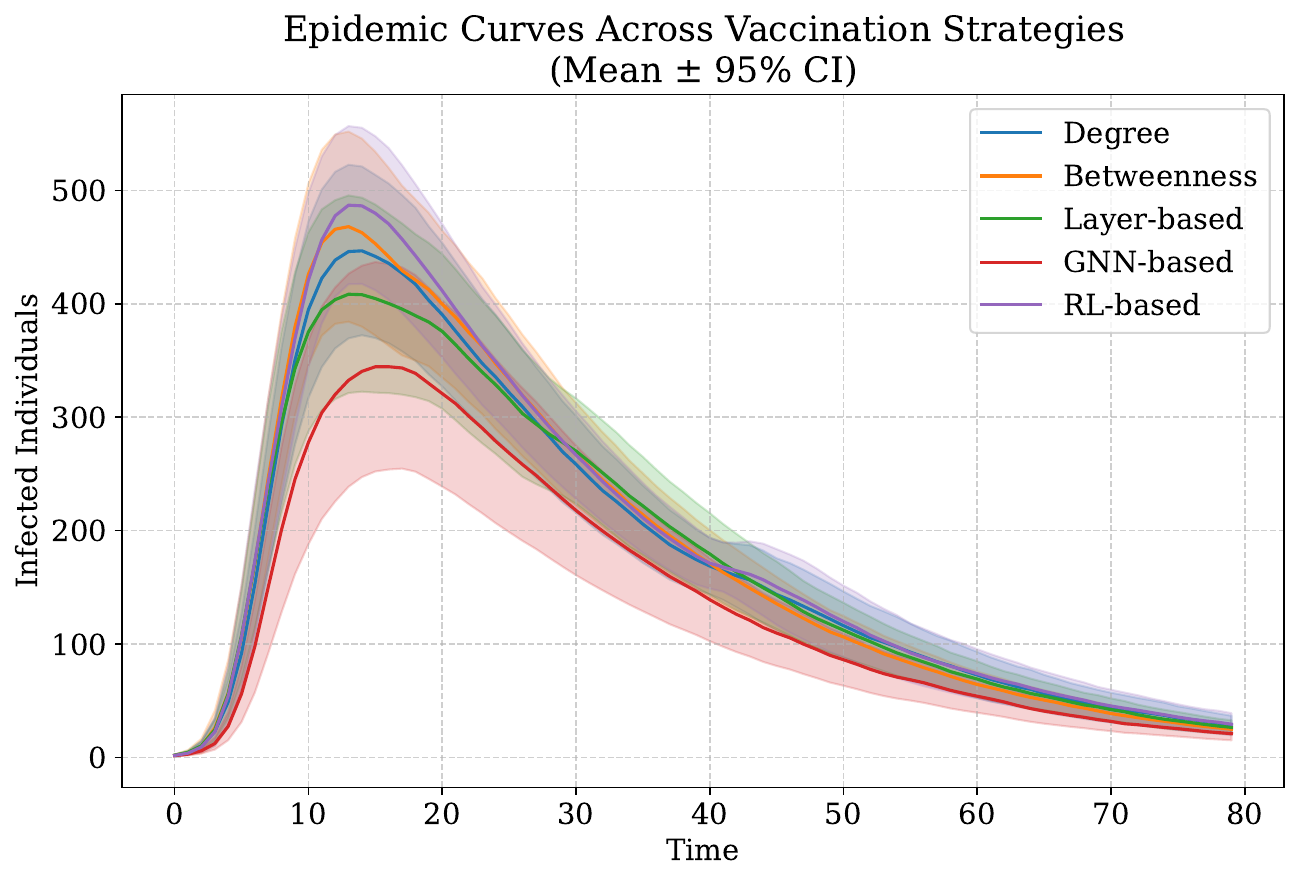}
\caption{The temporal evolution of infected individuals under each strategy is shown. Classical heuristics (degree, betweenness, layer-based) produce similar epidemic trajectories, while the GNN-based strategy achieves both the lowest peak and the fastest suppression. The RL-based baseline achieves the highest peak, as expected for an untrained policy. Shaded regions represent 95\% confidence intervals across 30 simulations.}
\label{fig:epicurves}
\end{figure}

\begin{figure}[t]
\centering
\includegraphics[width=\linewidth]{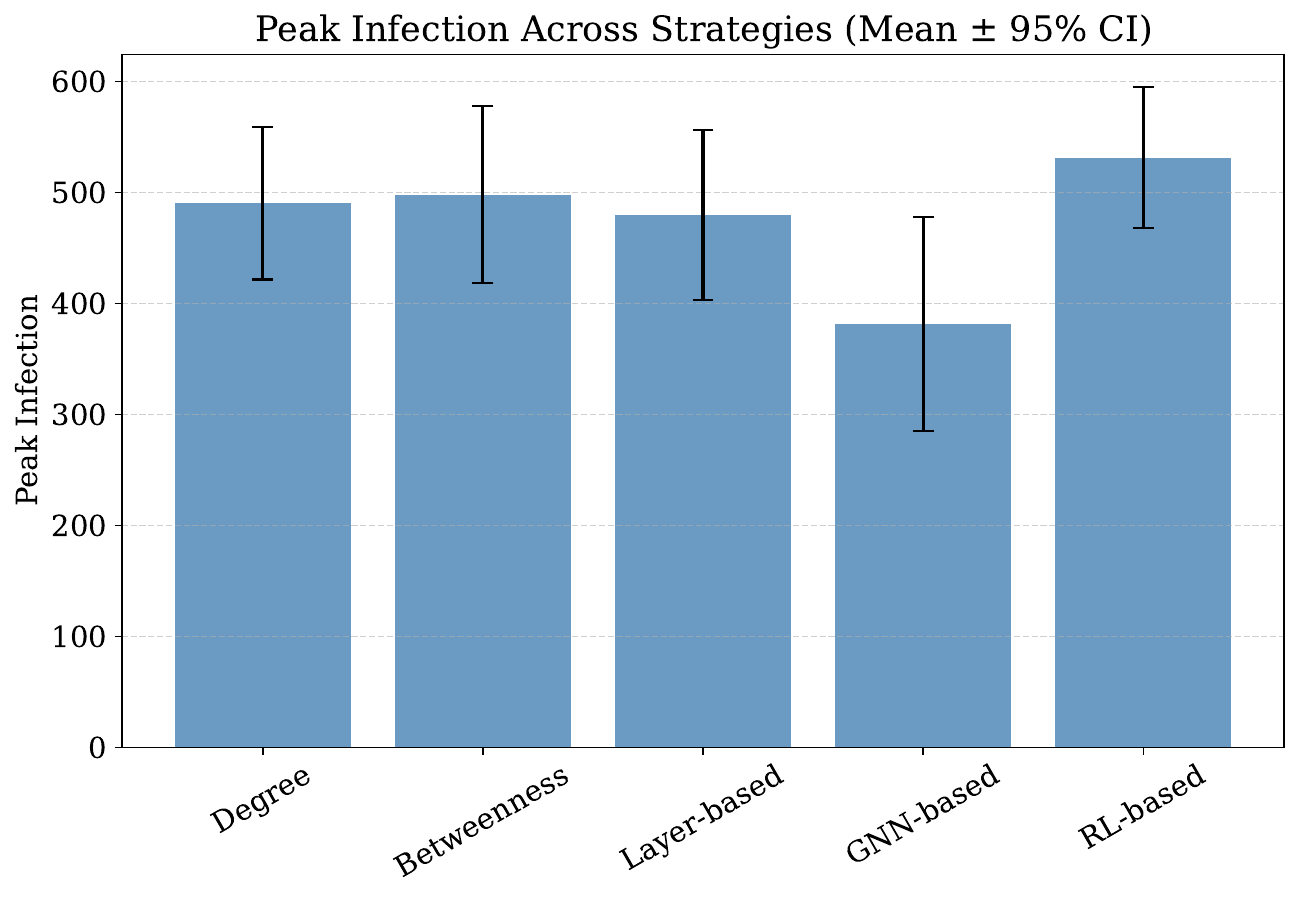}
\caption{ Degree-, betweenness-, and layer-based vaccination yield comparable peak infection levels, reflecting the network's dense and weakly modular structure. The GNN-based strategy substantially reduces the peak burden, demonstrating its ability to identify structurally critical nodes. The RL-based baseline performs have the highest levels of infection.}
\label{fig:peak}
\end{figure}

\begin{figure}[t]
\centering
\includegraphics[width=\linewidth]{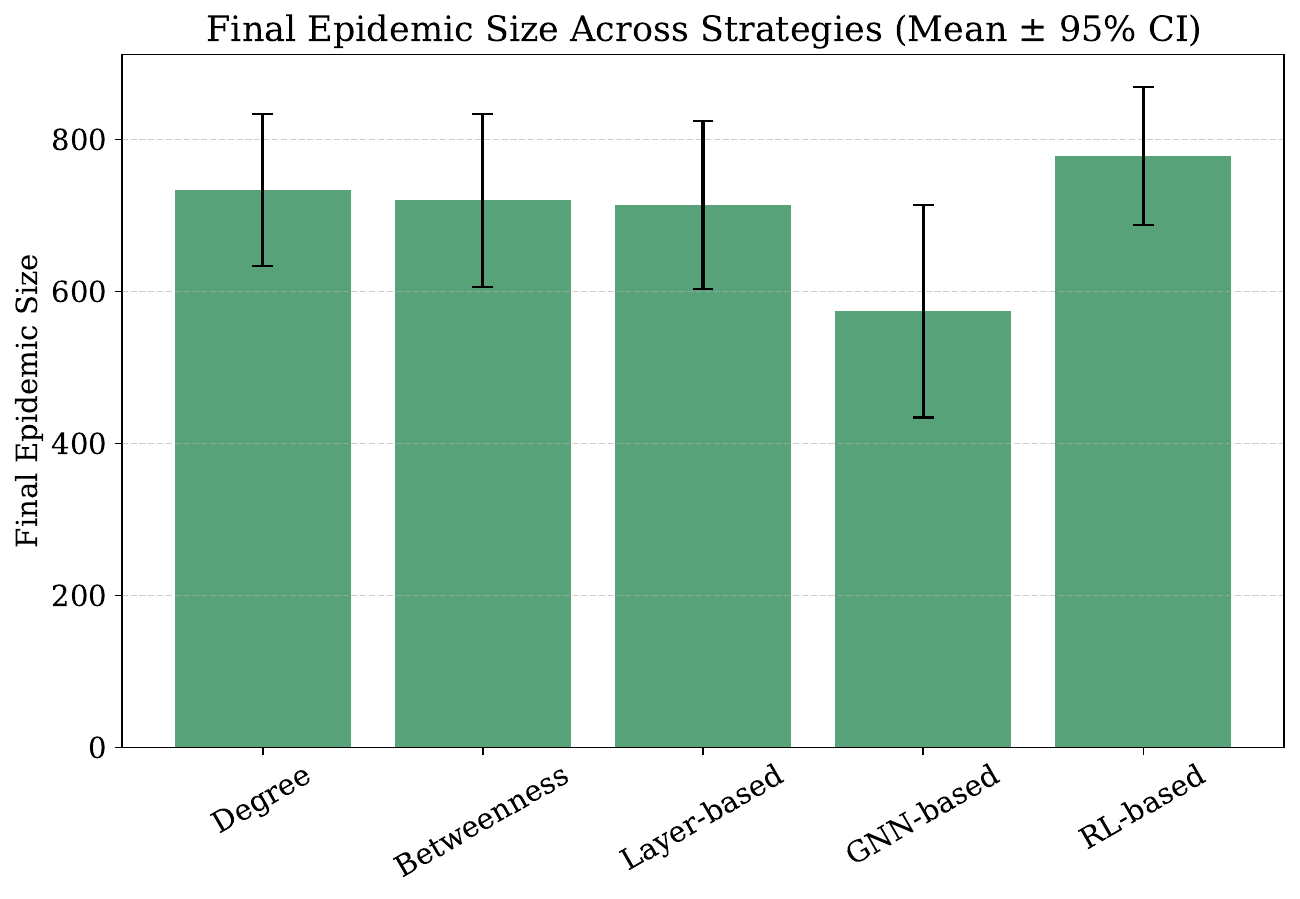}
\caption{ Classical heuristics again cluster tightly, while the GNN-based strategy achieves the smallest final epidemic size, preventing large-scale outbreaks. The RL-based baseline yields the largest epidemic size.}
\label{fig:finalsize}
\end{figure}

\begin{figure}[t]
\centering
\includegraphics[width=\linewidth]{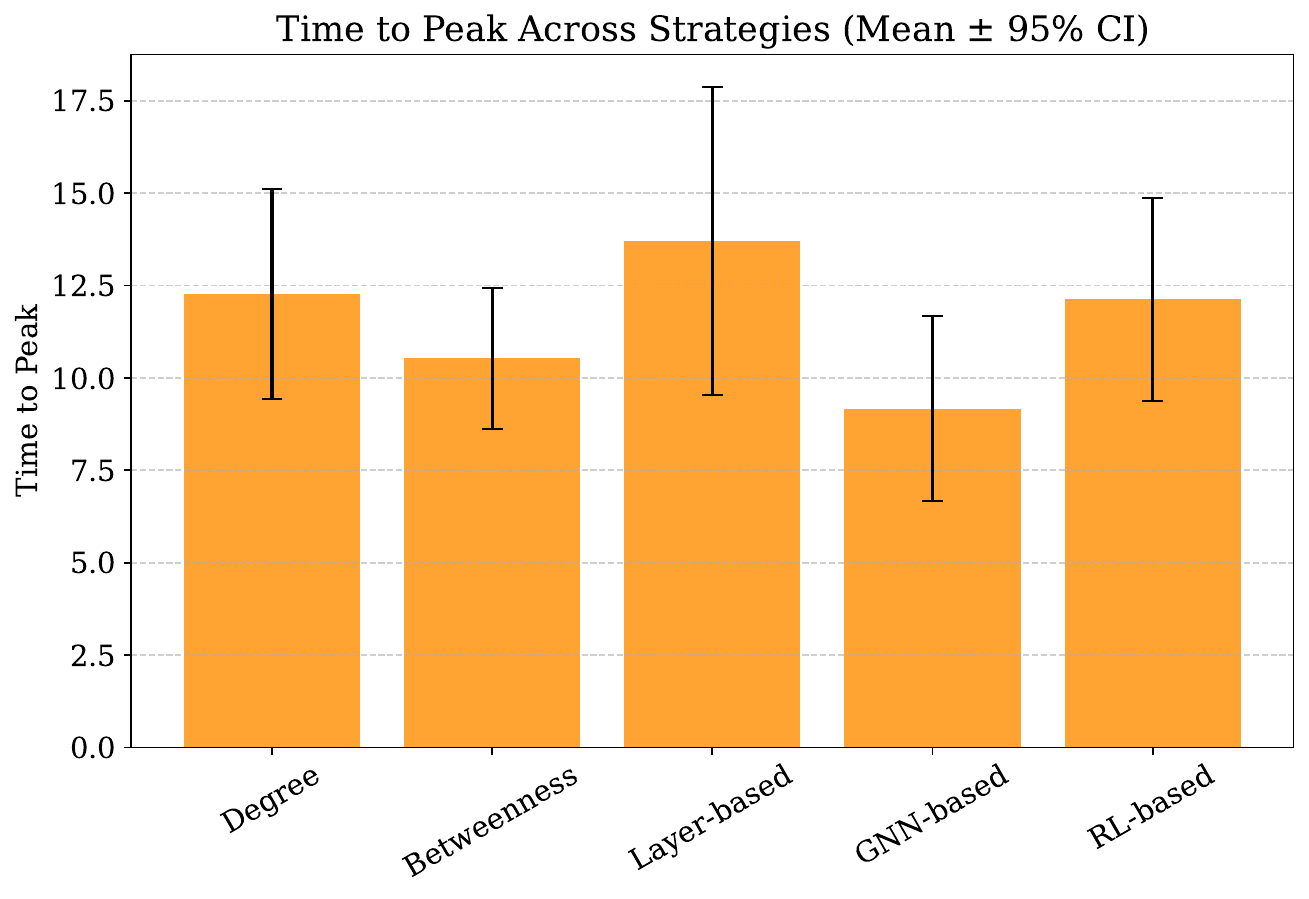}
\caption{The GNN-based strategy achieves the earliest time to peak infection, indicating the fastest epidemic suppression. 
Betweenness-based vaccination also reaches the peak earlier than degree- and RL-based strategies, reflecting its disruption of bridging nodes. Degree- and RL-based strategies exhibit similar times to peak, while the layer-based strategy peaks latest, consistent with its focus on community structure and slower cross-community spread.}
\label{fig:timepeak}
\end{figure}

\section{Conclusion}

This study investigated the effectiveness of classical and machine–learning-based vaccination strategies on a real-world contact network, using peak infection, final epidemic size, and time to peak infection as key epidemiological indicators. Across all metrics, the results demonstrate that classical heuristics—degree-based, betweenness-based, and layer-based vaccination—exhibit broadly similar performance. This clustering reflects the structural properties of the Email-Eu-Core network, whose dense connectivity and modest community structure reduce the discriminative power of traditional centrality measures. These findings are consistent with theoretical expectations that real-world networks often blur the distinctions predicted by idealized graph-theoretic models.

In contrast, the GNN-based vaccination strategy consistently outperformed all classical heuristics. It achieved the lowest peak infection, the smallest final epidemic size, and the earliest time to peak, indicating both stronger suppression of transmission pathways and faster epidemic decay. These improvements highlight the ability of graph neural networks to integrate multi-hop relational information and identify structurally critical nodes that are not captured by conventional centrality metrics. The GNN’s superior performance provides strong empirical support for the central claim of this work: machine learning can extend and refine classical network epidemiology by uncovering non-intuitive structural vulnerabilities.

The RL-based baseline, while currently untrained, performed as expected and establishes a lower bound for reinforcement learning approaches. Future work will incorporate fully trained RL agents capable of dynamic, temporally adaptive vaccination decisions, which may further enhance performance and complement the strengths of GNN-based strategies.

Overall, the results demonstrate that machine learning offers a powerful and flexible framework for epidemic control on complex networks. By outperforming classical heuristics across multiple epidemiological metrics, the GNN-based strategy illustrates the potential of data-driven approaches to inform public health interventions, especially in settings where network structure is heterogeneous, multi-layered, or only partially observable. These findings open the door to a new generation of vaccination strategies that combine theoretical guarantees with learned structural insights, bridging the gap between classical network science and modern machine learning.

\clearpage
\section*{Author contributions} MOO, BAA wrote the main manuscript text and prepared the main equations, figures, checked the writing errors in the manuscript, supervised the writing of the manuscript, and reviewed the manuscript.

\section*{Funding statement}
No funding was received for this research

\section*{Data availability}
The dataset on Email-Eu-Core contact network is available at
\url{https://github.com/arbenson/HigherOrderClustering.jl}.

\section*{Declarations conflict of interest} 
The authors declare no competing interests.

\bibliography{evo}

\end{document}